\documentclass[journal,twoside,web]{IEEEtran}
\usepackage{graphicx}
\usepackage{epstopdf}

\usepackage{cite}

\ifCLASSINFOpdf
\else
\fi
\usepackage[font={small}]{caption}
\usepackage{amsmath,amssymb}
\usepackage{amsfonts}
\usepackage{algorithm}
\usepackage{algorithmic}
\usepackage{array}
\usepackage{mathrsfs}
\usepackage{amsmath,amssymb,amsthm}
\usepackage{color}
\usepackage{graphicx}
\usepackage{textcomp}
\usepackage{xcolor}
\usepackage{algorithm}
\usepackage{amssymb}
\usepackage{tabularx}
\usepackage{booktabs}
\usepackage{cite}
\usepackage{graphicx}
\usepackage{epstopdf}
\usepackage{upgreek}
\usepackage{comment}
\usepackage{subfigure}
\usepackage{hyperref}
\usepackage{amsmath,amssymb}
\usepackage{subcaption}
\usepackage{url}
\usepackage{cleveref}
\newtheorem{assumption}{Assumption}
\newtheorem{theorem}{Theorem}
\newtheorem{definition}{Definition}

\newtheorem{remark}{Remark}
\newtheorem{proposition}{Proposition}

\crefname{assumption}{assumption}{Assumptions}

\begin{document}

\title{Safety-Critical Control with Offline-Online Neural Network Inference
}
\author{Junhui Zhang, \IEEEmembership{Member, IEEE}, Sze Zheng Yong, \IEEEmembership{Member, IEEE} and Dimitra Panagou, \IEEEmembership{Senior Member, IEEE}
	\thanks{This work was partially sponsored by the Office of Naval Research (ONR), under grant number N00014-20-1-2395 and NSF grant CNS-2312007. The views and conclusions contained herein are those of the authors only and should not be interpreted as representing those of ONR, the U.S. Navy, NSF or the U.S. Government.
}
  \thanks{Junhui Zhang is with the Department of Robotics at University of Michigan, Ann Arbor, MI 48109, USA (e-mail: junhuiz@umich.edu)}
\thanks{Sze Zheng Yong is with the Department of Mechanical and Industrial Engineering, Northeastern University, Boston, MA 02115, USA (e-mail: s.yong@northeastern.edu).}
 \thanks{Dimitra Panagou is with the Department of Robotics and the Department of Aerospace Engineering at University of Michigan, Ann Arbor, MI 48109, USA (e-mail: dpanagou@umich.edu).}}
\maketitle
\begin{abstract}

This paper presents a safety-critical control framework for an ego agent moving among other agents. The approach infers the dynamics of the other agents, and incorporates the inferred quantities into the design of control barrier function (CBF)-based controllers for the ego agent. The inference method combines offline and online learning with radial basis function neural networks (RBFNNs). The RBFNNs are initially trained offline using collected datasets. 
To enhance the generalization of the RBFNNs, the weights are then updated online with new observations, without requiring persistent excitation conditions in order to enhance the applicability of the method. Additionally, we employ adaptive conformal prediction to quantify the estimation error of the RBFNNs for the other agents' dynamics, generating prediction sets to cover the true value with high probability. Finally, we formulate a CBF-based controller for the ego agent to guarantee safety with the desired confidence level by accounting for the prediction sets of other agents' dynamics in the sampled-data CBF conditions. Simulation results are provided to demonstrate the effectiveness of the proposed method.

\end{abstract}


\section{Introduction} 
Ensuring safe motion for autonomous agents is an ongoing and challenging task. 
Control Barrier Functions (CBFs) \cite{ames2016control,garg2024advances,zhang2023novel} enforce the satisfaction of safety constraints along the system trajectories based on conditions that rely on the system model. Earlier work assumes that the dynamics of other agents are known within certain bounds \cite{mustafa2022adversary,parwana2022trust}, which may be limiting or conservative for real-world applications. 
In this work, we develop an algorithm to infer the other agents' dynamics and integrate the learned dynamics into sampled data CBF designs.

Inferring the dynamics of other agents can be seen as a system identification problem. 
Unknown dynamics can be learned via non-parametric, data-driven
approaches 
\cite{jin2023robust}, or via parametric methods 
\cite{lopez2020robust}. 
When there is no prior knowledge about the structure of the unknown dynamics, a three-layer neural network known as radial basis function neural network (RBFNN) can be employed \cite{lee2004adaptive} with the advantage that the universal approximation theorem \cite{park1991universal} ensures that any continuous function can be approximated by a RBFNN arbitrarily closely when the number of neurons is sufficiently large. 
To train the RBFNN, one can use either offline or online approaches \cite{pazouki2015efficient,seghouane2019adaptive}. For offline learning methods, 
generalizing to unseen data and behaviors during the online implementation process is a challenge. 
On the other hand, training RBFNNs online poses challenges in achieving convergence of parameters, especially without the persistent excitation (PE) condition 
\cite{zheng2017relationship}. 

Our first aim in this work 
is to develop an offline-online inference learning method to infer other agents' dynamics. We use collected data to train the RBFNNs offline, and then the weights are updated during the online process. The offline-trained RBFNNs offer a ``warm start" for the online-learning process, ensuring that the estimation error remains relatively small before sufficient data are collected for training the RBFNNs  online. On the other hand, the online-learning process enhances the generalization of offline trained RBFNNs since it keeps tuning their weights using new observations. 

Our next aim is to quantify the estimation error of the learned dynamics so that we can use them in CBF-based control. Conformal prediction is a promising method to quantify uncertainty for arbitrary prediction algorithms \cite{angelopoulos2023conformal}.
In \cite{gibbs2021adaptive}, adaptive conformal prediction (ACP) is proposed that does not rely on exchangeability of data points, although the confidence is only guaranteed in an average sense. This motivates us to employ ACP in the estimation error quantification of our offline-online neural network inference. Specifically, prediction sets are produced online by ACP such that the true value is covered within the prediction sets with a desired probability. 
Note that in \cite{dixit2023adaptive},
prediction sets are embedded into a model predictive control framework, while in the current work, the prediction sets are utilized within a sampled-data CBF design.

The main contributions of this work are summarized here: (1) We propose an offline-online method that utilizes RBFNNs to 
infer the unknown dynamics of other agents. 
(2) We leverage adaptive conformal prediction to online quantify the estimation error of our offline-online trained RBFNNs. (3) We incorporate prediction sets produced by adaptive conformal prediction into a sampled-data CBF framework, ensuring that the ego agent maintains safety with a desired confidence level.

\textrm{Notations:} The 
set of real numbers and non-negative real numbers are denoted by $\mathbb{R}$ and $\mathbb{R}^+$, respectively. 
A continuous function $\beta: (-b,a) \rightarrow (-\infty,+\infty)$ is 
an extended class $\mathcal{K}$ function for $a, b\in \mathbb{R}^{+}$, denoted by $\beta \in \mathcal{K}$, if it is strictly increasing and $\beta(0)=0$. $\|\cdot\|$ and $\|\cdot\|_{\infty}$ represent 2- and infinity-norm, respectively. For $a\in \mathbb{R}$, $\lceil a\rceil$ is the smallest integer greater than or equal to $a$. $\text{Prob}[A]$ is the probability of an event $A$. $L_fh(x)$ is the Lie derivative of $h$ along $f$ at $x$. 

\section{Preliminaries and Problem Formulation}
\subsubsection{Preliminaries}
Consider a control affine system,
\begin{align}
\label{eq:system}
\dot{x}=f(x)+g(x)u,
\end{align}
where $x\in \mathcal{X}\subset \mathbb{R}^{n}$ and $u\in \mathcal{U} \subset \mathbb{R}^{m}$ are the state and control vectors of \eqref{eq:system}. $f:\mathcal{X}\rightarrow \mathbb{R}^{n}$ and $g:\mathcal{X}\rightarrow \mathbb{R}^{n\times m}$ are Lipschitz-continuous functions.  

\begin{definition}\cite{ames2016control}
Let a safety set $S$, i.e., the set of safe states,  for the system \eqref{eq:system} be defined as $S=\{x\in \mathcal{X}:h(x)\geq 0\}$,
where $h: \mathcal{X}\rightarrow \mathbb{R}$ is a continuously differentiable function. The function $h$ is a control barrier function for system \eqref{eq:system} on the set $S$ if there exists $\beta\in \mathcal{K}$, such that
\begin{align*}
	\sup \limits_{u\in \mathbb R^m}\{L_fh(x)+L_gh(x)u\}\geq -\beta(h(x)), \quad \forall x\in S.
\end{align*}
\end{definition}

\subsubsection{Problem Formulation}
In this work, we consider an ego agent with dynamics described by
\begin{align}
\label{ego-agent-model}
\dot{x}_{e}=f(x_{e})+g(x_{e})u_{e},
\end{align}
where $x_{e}\in \mathcal{X}_e\subset \mathbb{R}^{n_e} $ and $u_{e}\in \mathcal{U}_e \subset \mathbb{R}^{m_e}$ are the state and control input vectors of ego agent, respectively. $f:\mathcal{X}_e\rightarrow \mathbb{R}^{n_e}$ and $g:\mathcal{X}_e\rightarrow\mathbb{R}^{n_e\times m_e}$ are Lipschitz-continuous functions. The set of other agents is denoted by $\mathcal{V}=\{1,\cdots,\mathcal{N}\}$.  
Let $x_i\in \mathcal{X}_i\subset\mathbb{R}^{n_i}$ denote the state of the $i$th agent in $\mathcal{V}$, and $\mathbf{x}=[x_{e}^T, x_1^T, x_2^T,\cdots, x_{\mathcal{N}}^T]\in \mathbb{R}^{n}, n=n_e+\sum^{\mathcal{N}}_{i=1} n_i$, which includes all the agents' states. Suppose that each agent $i\in \mathcal{V}$ observes $\mathbf{x}$ and designs its action based on the observations.  Then, the dynamics of agent $i$ are 
$\dot{x}_i=F_i(\mathbf{x})$, where $F_i$ is unknown to ego agent.

The ego agent is assigned to track its reference trajectory $x^r_{e}$ with reference control $u^r_{e}$, and avoid collisions with other agents. One can define functions $h_{i}:\mathbb{R}^{n_e}\times \mathbb{R}^{n_i}\rightarrow \mathbb{R}$ to encode the distance of ego agent w.r.t agent $i\in \mathcal{V}$, so that the set $S_{i}=\{x_e\in \mathcal{X}_e,  x_i\in \mathcal{X}_i: h_{i}(x_e,x_i)\geq 0\}$ encodes the safety set for the ego agent and agent $i$. The CBF condition of the ego agent w.r.t. an agent $i\in \mathcal V$ is 
\begin{align}
\label{CBF condition}
\begin{array}{rl}
\dot{h}_{i}	&=L_fh_i(x_e,x_i)+L_gh_i(x_e,x_i)u_e+\frac{\partial h_{i}}{\partial x_i}\dot{x}_i \\
 &\geq -\beta_{i}(h_{i}(x_e,x_i)),
 \end{array}
\end{align} 
 where $\beta_{i}\in \mathcal{K}$. 
 
The objectives of this work are: \textbf{Problem 1}: \textit{How can the ego agent infer other agents' dynamics?} 
\textbf{Problem 2}: \textit{How can the ego agent quantify the estimation error of the inference algorithm for other agents' dynamics?} \textbf{Problem 3}: \textit{How to design a safe controller for the ego agent incorporating inference of other agents, ensuring that the ego agent safely tracks its reference trajectory as closely as possible?}

\section{Methodology}
In our setup, the ego agent observes states of all agents and updates its control signal at time instants $t_k=k\Delta t, k\in\{0,1,2,\dots\}$, where $\Delta t>0$ is a fixed and sufficiently small sampling time interval. A zero-order hold (ZOH) control law is utilized so that $u_e(t)=u_e(t_k)$, for $t\in [t_{k},t_{k+1})$. More specifically, at each sampling time instant $t_k$, the ego agent observes states of all agents $\mathbf{x}(t_k)$, and needs to update its control signal $u_e(t_k)$ based on a CBF-based controller. 
In Section \ref{sec:inference}, we study how to estimate $\dot{x}_i(t_k)$ at sampling time $t_k$ using online observations and offline collected data so that this estimate can be used in the CBF-based controller. Moreover, we characterize in Section \ref{sec:acp} the estimation error with high average confidence using adaptive conformal prediction. Finally, since the satisfaction of \eqref{CBF condition} at sampling instants only is not sufficient for safety, we design CBF conditions 
for sampled-data systems \cite{breeden2021control} with the incorporation of estimates of other agents' dynamics and their estimation errors in Section \ref{sec:CBF}. This  framework is summarized in Fig. \ref{fig:1} and to facilitate this, we make the following assumptions.

\begin{figure}
	\begin{center}
\includegraphics[width=0.97\columnwidth]{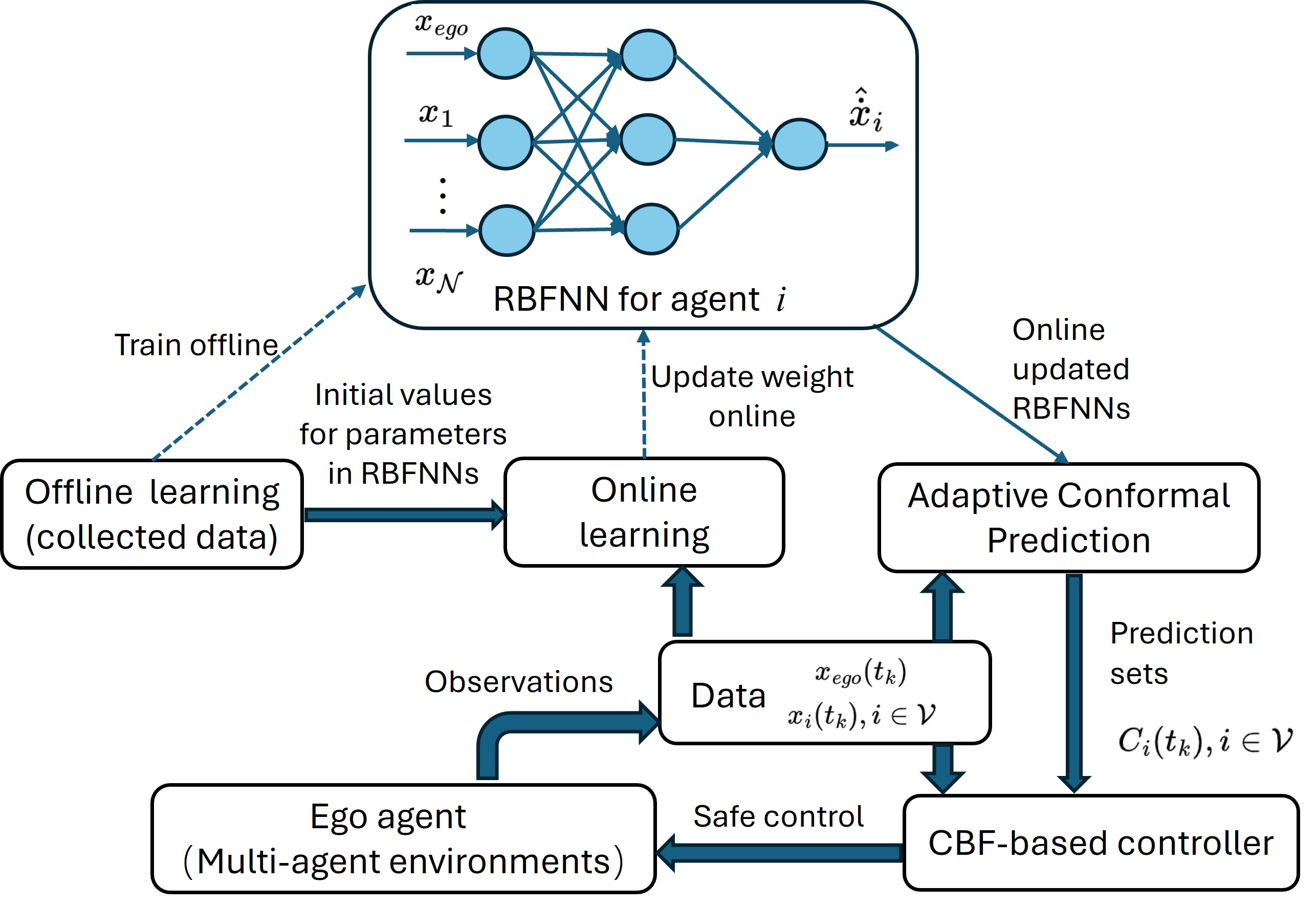}
		\hspace{200in}\vspace{-5mm}
		\caption{The proposed framework of sampled data CBF-based controller with offline-online neural network inference: RBFNNs are utilized to learn other agents' dynamics and are trained in an offline-online manner. Adaptive conformal prediction (ACP) quantifies the estimation error online, and the prediction sets based on ACP are incorporated into the sampled-data CBF-based controller.\label{fig:1}}
  \vspace{-20pt}
	\end{center}
\end{figure}

\begin{assumption}
\label{assp1}
For the ego agent, there exists $\delta_{e,k}\geq0$ such that $\|x_e(t)-x_e(t_k)\|\leq \delta_{e,k}$ for $\forall t\in[t_k,t_{k+1})$, i.e., the change of ego agent's state during the sampling time interval is bounded.\footnote{The bound $\delta_{e,k}$ can be quantified as in \cite{breeden2021control}, that is, $\|x_e(t)-x_e(t_k)\|\leq \sigma (t-t_k) \leq \sigma \Delta t$, where $\sigma =\sup_{x_{e}\in \mathcal{X}_e,u_e\in \mathcal{U}_e}\|f(x_e)+g(x_e)u_e\|$.}
\end{assumption}

\begin{assumption}
\label{assp2}
For agent $i\in \mathcal{V}$, there exist $\delta_{i,k}\geq0$ and $\eta_{i,k}\geq0$ such that $\|x_i(t)-x_i(t_k)\|\leq \delta_{i,k}$
for $\forall t\in[t_k,t_{k+1})$. 
 In addition, we assume that $F_i$ is Lipschitz-continuous, i.e., that $\|\dot{x}_i(t)-\dot{x}_i(t_k)\|=\|F_i(\mathbf{x}(t))-F_i(\mathbf{x}(t_k))\|\leq c_{F_i}\|\mathbf{x}(t)-\mathbf{x}(t_k)\|\leq\eta_{i,k}=c_{F_i}( \sum_j^\mathcal{N}\delta_{j,k}+\delta_{e,k})$, for $t\in[t_k,t_{k+1})$, where $c_{F_i}$ is a known Lipschitz constant for $F_i$.
 \end{assumption}
 
\begin{assumption} \label{assp3}
The historical dynamics of other agents (i.e., $\dot{x}_i$ with a one-step time delay) can be computed from measured state trajectory, e.g., using finite difference via 
$\dot{x}_i(t_{k-1})=\frac{x_i(t_k)-x_i(t_{k-1})}{\Delta t}$.
For simplicity, we assume the approximation error 
is zero. 
If desired, the approximation error can also be quantified and incorporated in a straightforward manner.
\end{assumption}

\subsection{Offline-Online Inference}\label{sec:inference}
First, we develop an offline-online framework to infer other agents' dynamics for addressing \textbf{Problem 1}.

\subsubsection{\textbf{Offline Learning}}

To collect a dataset offline, the ego agent moves in a multi-agent environment with any safe policy. The other agents act following their intentions in reaction to the ego agent. The states of all agents are observed and collected. This means we have $\mathbf{x}(t_k)$, where $k \in \{0,1,\cdots,N\}$. The corresponding dynamics of other agents are also obtained by Assumption \ref{assp3}, that is, $\dot{x}_i(t_k)$, for all $k \in \{0,1,\cdots,N-1\}$.
 
Using the offline dataset, we first estimate the 
functional relationship $F_i(\mathbf{x})$ between the dynamics of other agents $\dot{x}_i$ and the states $\mathbf{x}$. From the well-known Universal Approximation Theorem \cite{park1991universal}, when the number of neurons is sufficiently large, the radial basis function neural network (RBFNN) can approximate any continuous function arbitrarily closely; thus, we use a RBFNN to 
learn the functional relationship $F_i(\mathbf{x})$ as
\begin{align}
\label{RBFNN}
\hat{\dot{x}}_i=\hat{W}_i^{T}\Phi_i(\mathbf{x}),
\end{align}
where $\hat{\dot{x}}_i$ is the estimate for $\dot{x}_i$, $\hat{W}_i\in \mathbb{R}^{(M_i+1) \times n_i}$ is the weight. $\Phi_i(\mathbf{x})=\begin{bmatrix}e^{\frac{-\left|\mathbf{x}-c_{i,1}\right|^2}{2\rho_{i,1}^2}},\cdots, e^{\frac{-\left|\mathbf{x}-c_{i,M}\right|^2}{2\rho_{i,M_i}^2}},1\end{bmatrix}^T$ is the Gaussian basis function with centers $C_i=[c_{i,1},\dots,c_{i,M_i}]\in \mathbb{R}^{n \times M_i}$ and widths $\rho_i=[\rho_{i,1},\dots,\rho_{i,M_i}]^T\in \mathbb{R}^{M_i}$. 

The goal of training the RBFNN is to obtain the optimal weights such that the average of the estimation errors $\epsilon_i(t_k)=\hat{\dot{x}}_i(t_k)-\dot{x}_i(t_k)$ across all samples is minimized, i.e., 
\begin{align}
\label{offline optimization}
	\min_{W_i} \frac{1}{N} \sum_{k=0}^{N-1} \| \epsilon_i(t_k)\|^2.
\end{align} 

The RBFNN can be trained offline with methods such as gradient-based methods \cite{han2021accelerated} and the orthogonal least-squares methods \cite{chen1991orthogonal}. The centers in the basis functions can be optimized with weights together during the training process or determined by implementing clustering approaches on collected datasets, e.g., the $K$-means method.\footnote{Note that offline training occurs before the system deployment, so it is acceptable for it to take a longer time, if required.} 
{The widths of radial basis functions are manually tuned, e.g.,  close to 1.} 

\subsubsection{\textbf{Online Learning}} 
 
The trained RBFNNs may not accurately estimate behaviors of the other agents that do not belong to the collected dataset. To enhance the generalization of RBFNNs, the weights obtained from offline training are utilized as a warm initialization $\hat{W}_i(t_0)$ and updated with new observations during the online process.
For the online learning algorithm, the convergence of weights usually relies on the conditions of Persistent Excitation (PE), which are challenging to check online. In \cite{chowdhary2010concurrent}, the concurrent learning method was proposed, guaranteeing convergence without PE conditions, provided that the recorded data (either offline or online) contains as many linearly independent elements as the dimensions of the basis function.
The key idea is that the weights are updated using both instantaneous observed data and the recorded data $D_{z,i}=\{\mathbf{x}(m) \;| \; m=1,\cdots,\xi_i\}$ that has as many linearly independent elements as the dimension of the basis function $\Phi_i$, when it becomes available, as follows:
\begin{align}
\label{W_Update law}
\begin{array}{l}
\hat{W}_i(t_{k})=\\ \hspace{-0.1cm}
\begin{cases}\hat{W}_i(t_{k-1})\hspace{-0.05cm}-\hspace{-0.05cm}\Gamma_1 \Phi_i(\mathbf{x}(t_{k-1}))\epsilon_i(t_{k-1})\Delta t, & \hspace{-0.2cm}\textrm{rank}(Z_i)\hspace{-0.05cm}<\hspace{-0.05cm}M_i\hspace{-0.05cm}+\hspace{-0.05cm}1,\\
\begin{array}{l}\hat{W}_i(t_{k-1})-[\Gamma_1 \Phi_i(\mathbf{x}(t_{k-1}))\epsilon_i(t_{k-1})\\
 +\sum^{\xi_i}_{m=1} \Gamma_2 \Phi_i(\mathbf{x}(m))\epsilon_i(m)]\Delta t,
\end{array} & \hspace{-0.2cm}\text{otherwise},
\end{cases}
\end{array}
\end{align}
where $Z_i=[\Phi_i(\mathbf{x}(1)),\cdots,\Phi_i(\mathbf{x}(\xi_i))]$, the estimation error $\epsilon_i(t_{k-1})$ is defined in above \eqref{offline optimization}, $\epsilon_i(m)$ is the estimation error for the dynamics of agent $i$ at the $m$th data point of the recorded data $D_{z,i}$, and $\Gamma_1$ and $\Gamma_2$ are positive constants. 

In the above, we adopt a standard adaptive law from the adaptive control literature  before we collect enough data for $D_{z,i}$, that is, when $\textrm{rank}(Z_i)<M_i+1$, $Z_i=[\Phi_i(\mathbf{x}(1)),\cdots,\Phi_i(\mathbf{x}(\xi_i))]$ and switch to concurrent learning when $\textrm{rank}(Z_i)=M_i+1$. As discussed in \cite[Theorem 1]{djaneye2019gradient}, when the recorded data satisfies the above rank condition, the weights converge to a small bounded set around its ideal value. 

\subsection{Adaptive Conformal Prediction} \label{sec:acp} 
Next, we aim to quantify the estimation error of the proposed offline-online inference method and solve \textbf{Problem 2}.

At sampling time $t_k$, the ego agent observes the state $\mathbf{x}(t_{k})$, and the other agents' dynamics at $t_{k-1}$ are calculated by Assumption \ref{assp3}. The weight estimation $\hat{W}_i$ of RBFNN is updated online under 
\eqref{W_Update law} with data point $(\mathbf{x}(t_{k-1}),\dot{x}(t_{k-1}))$ at $t_{k-1}$ and recorded data $D_z$. Subsequently, the estimate $\hat{\dot{x}}_i(t_k)$ of other agents' dynamics at time $t_k$ are obtained using the RBFNN \eqref{RBFNN}; however, no estimation error in the prediction of the RBFNN is available. Conformal prediction determines a prediction set that covers the true value with a desired probability. The original conformal prediction assumes exchangeability of data points. However, in our problem, 
the state $\mathbf{x}(t_k)$ depends on the state $\mathbf{x}(t_{k-1})$ and agents' actions at $t_{k-1}$, violating the exchangeability condition. In \cite{gibbs2021adaptive}, an adaptive conformal prediction (ACP) method was further proposed that does not rely on the exchangability condition, 
    with coverage confidence guarantees in an average sense. We employ the ACP to quantify the estimation error of dynamics at $t_k$ estimated by the RBFNN.

Let $D_{c,i}(t_k)=\{(\mathbf{x}(t_{k-(L+1-j)}), \dot{x}_i(t_{k-(L+1-j)})), j\in\{1,2,\dots,L\} \}$ denote the set consisting of the most recent $L$ data points with observed states and calculated historical dynamics by Assumption \ref{assp3} at the sampling time $t_k$. At the subsequent sampling time, the oldest sample is removed, and a new sample is added to the set. In the online conformal prediction framework, $D_{c,i}(t_k)$ acts as the calibration dataset. That is, the most recent $L$ data points are utilized to predict the estimation error for time $t_k$.

After updating the weight $\hat{W}_i$ using \eqref{W_Update law} 
at time $t_k$, the nonconformity scores, i.e., estimation errors of the RBFNN on the calibration dataset are calculated as:
\begin{align}
R_{i,j}=\|\dot{x}_i(t_{k-(L+1-j)})-\hat{W}_i^{T}\Phi_i(\mathbf{x}(t_{k-(L+1-j)}))\|_{\infty}, 
\end{align}
where $j\in \{1,2,\cdots,L\}$. 
Due to the dependence between time-series data points, inspired by ACP in \cite{gibbs2021adaptive}, we use a recursively updated failure probability $\alpha_i(t_k)$ instead of consistently using a fixed target failure probability $\alpha_i \in (0,1)$. From \cite{dixit2023adaptive}, the width of the prediction set for $\hat{\dot{x}}_i(t_k)$ is the $(1-\alpha_i(t_k))$th quantile of the sequence of $R_{i,1},\cdots,R_{i,L},\infty$, and when $\alpha_i(t_k)\in [\frac{1}{L+1}, 1)$, the width can be calculated by:
\begin{align}
\label{Calculate width}
 q_i(t_k)=\lceil(1-\alpha_i(t_k))(L+1)\rceil\text{th smallest of } \{R_{i,m}\}_{m=1}^L,
\end{align}
with some special cases discussed in Remark 1. 
Then, the prediction set $C_{i}(t_k)$ is:
\begin{align}\label{eq:calibrationset}
C_{i}(t_k)=\{\dot{x}_i \; | \; \|\dot{x}_i-\hat{\dot{x}}_i(t_k)\|_{\infty}\leq q_i(t_k)\}. 
\end{align}
For the adaptive law of miscoverage level $\alpha_i(t_k)$, we define $err_i(t_{k-1})=1$ if $\dot{x}_i(t_{k-1})\notin C_{i}(t_{k-1})$, and $err_i(t_{k-1})=0$ otherwise, where $\dot{x}_i(t_{k-1})$ is calculated by Assumption \ref{assp3}. Then, the $\alpha_i(t_k)$ increases if the prediction set successfully covers the true value at the previous time step and decreases otherwise. The update law of $\alpha_i(t_k)$ is given by 
\begin{align}
\label{Update law of alpha_t}
\alpha_i(t_k)=\alpha_i(t_{k-1})+\gamma(\alpha_i-err_i(t_{k-1})),
\end{align}
where $\gamma$ is the learning rate. Then, from \cite{gibbs2021adaptive, dixit2023adaptive}, the prediction set in \eqref{eq:calibrationset} is guaranteed to cover the true dynamics of other agents with high average probability:

\begin{proposition}\cite{gibbs2021adaptive, dixit2023adaptive}
\label{prop1}
Consider a time horizon $[t_0,T]$, with sampling time instants $\{t_0, t_1, t_2, \cdots, t_K\}$. For each sampling time $t_k$, where $k\in\{1,2,\cdots,K\}$, the adaptive failure probability $\alpha_i(t_k)$ is 
updated using recursion \eqref{Update law of alpha_t} with a learning rate $\gamma>0$, a fixed target failure probability $\alpha_i\in(0,1)$ and an initial value $\alpha_i(t_0)$,
then, it holds that  
\begin{align}
\label{Probability}
\textstyle\frac{1}{K}\sum^{K-1}_{k=0}\text{Prob}[\dot{x}_i(t_k)\in C_i(t_k)]\geq 1-\alpha_i-\mathcal{P}_i(K),
\end{align}
    where $\mathcal{P}_i(K)=\frac{\alpha_i(t_0)+\gamma}{K\gamma}$, and  $\lim_{K\to +\infty}\mathcal{P}_i(K)=0$.
\end{proposition}

\begin{remark}\label{rem:1}
From the recursion \eqref{Update law of alpha_t}, it is possible for $\alpha_i(t_k)$ to fall outside the range $[\frac{1}{L+1}, 1)$ during the online process, causing the \eqref{Calculate width} to fail in
finding the widths of prediction sets in these instants. 
To address this, we enforce $q_i(t_k)=R_{i,\max}$ when $\alpha_i(t_k)<\frac{1}{L+1}$, ensuring a sufficiently large prediction set to definitely cover the true value. The value of $R_{i,\max}$ can be determined by the potential maximum value of the dynamics of agent $i$ in real-world applications. For example, if the dynamics represent the speed of an agent, $R_{i,\max}$ can be set to the maximum speed limit of agent $i$. Similarly, we enforce $q_i(t_k)=R_{i,\min}=0$ when $\alpha_i(t_k)\geq1$. 
\end{remark}

\begin{remark} Note that the result in Proposition \ref{prop1} is independent of the calibration dataset used, including the choice of a moving horizon of length $L$ in Section \ref{sec:acp}. This can be observed from scrutinizing its proof in \cite[Lemma 4.1]{gibbs2021adaptive} that holds as long as $err_i(t_{k-1})$ is binary, regardless of the choice of calibration set that determines it. However, the choice of calibration dataset does affect the width of the prediction sets and how to optimally select this with a constraint on the calibration dataset size is a subject of ongoing research. 
\end{remark}

\subsection{CBF-based Controller with Inference}\label{sec:CBF}
To solve \textbf{Problem 3}, given the prediction set at $t_k$ by ACP, we can further find a prediction set for all $t\in [t_k,t_{k+1})$. By Assumption \ref{assp2}, the change of agent $i$'s dynamics during the sampling time interval is bounded, that is, $\|\dot{x}_i(t)-\dot{x}_i(t_k)\|\leq \eta_{i,k}$, for $t\in[t_k,t_{k+1})$. Then, from \eqref{eq:calibrationset}, we can define 
\begin{align}
\mathcal{C}_{i,k}=\{\dot{x}_i \;| \; \|\dot{x}_i-\hat{\dot{x}}_i(t_k)\|_{\infty}\leq q_i(t_k)+\eta_{i,k}\},
\end{align}
to ensure that $\dot{x}_i(t)\in \mathcal{C}_{i,k}, \forall t\in [t_k,t_{k+1})$, if $\dot{x}_i(t_k)\in C_{i}(t_k)$. 

Using this, we derive sampled-data CBF conditions to guarantee safety with high average probability.

\begin{theorem}
\label{theorem1} Consider a time horizon $[t_0,T]$ with sampling time instants $\{t_0, t_1, t_2, \cdots, t_K\}$ and  safety set $S_i$ for $i\in \mathcal{V}$, let $c_{f,i}$, $c_{g,i}$ and $c_{\beta,i}$ be Lipschitz constants for $L_fh_i$, $L_gh_i$ and $\beta_i(h_i)$ with respect to the 2-norm, respectively. Under \cref{assp1,assp2,assp3}, if the ZOH control of ego agent  at each $t_k\in\{t_0, t_1, t_2, \cdots, t_K\}$ satisfies the following sampled-data CBF conditions:
\begin{align}
\label{Sampling-data CBF}
\begin{array}{r}L_fh_i(x_e(t_k),x_i(t_k))+L_gh_i(x_e(t_k),x_i(t_k))u_{e}(t_k)\\+M_{i,k}\geq 
-\beta_{i}(h_{i}(x_e(t_k),x_i(t_k)))+\phi_{i,k},
\end{array}
\end{align}
with $\phi_{i,k}=(c_{f,i}+c_{g,i}u_{e,\max}+c_{\beta,i})(\delta_{e,k}+\delta_{i,k})$, $u_{e,\max}=\min_{u_e \in \mathcal{U}_e, \|u_e\|\le \delta} \delta$ and 
\begin{align}
\label{Define M}
\textstyle M_{i,k}=\min_{{\dot{x}}_i\in \mathcal{C}_{i,k}}\{\frac{\partial h_{i}}{\partial x_i}{\dot{x}}_i\}=\sum_{r=1}^{n_i}b_{i,k,r},
\end{align}
with $b_{i,k,r}=\min\{\mu_{i,r}(\hat{\dot{x}}_{i,r}(t_k)-q_i(t_k)-\eta_{i,k}),\mu_{i,r}(\hat{\dot{x}}_{i,r}(t_k)+q_i(t_k)+\eta_{i,k})\}$, where $\mu_{i,r}$ is the $r$th column of ${\frac{\partial h_{i}}{\partial x_i}}$ and $\hat{\dot{x}}_{i,r}(t_k)$ is the $r$th row of $\hat{\dot{x}}_{i}(t_k)$, 
then safety is guaranteed with a probability of at least $1-\alpha_i-\mathcal{P}_i(K)$ on average, i.e., 
\begin{align}
\begin{array}{r}
\frac{1}{K}\sum^{K-1}_{k=0}\text{Prob} [(x_e(t),x_i(t))\in S_{i}, \forall t\in [t_k,t_{k+1})]\\
\geq 1-\alpha_i-\mathcal{P}_i(K),
\end{array}
\end{align}
with $\mathcal{P}_i(K)=\frac{\alpha_i(t_0)+\gamma}{K\gamma}$ and $\lim_{K \to +\infty}\mathcal{P}_i(K)=0$.
\end{theorem}
\begin{proof}
If $\dot{x}_i(t_k)\in C_i(t_k)$, $k<K$,   then for all $t\in [t_k, t_{k+1})$, it holds that,
\begin{align*}
\begin{array}{l}
L_fh_i(x_e,x_i)+L_gh_i(x_e,x_i)u_e(t_k)+\frac{\partial h_{i}}{\partial x_i}\dot{x}_i\\
\stackrel{\text{\eqref{Define M}}}{\geq} L_fh_i(x_e(t_k),x_i(t_k))+L_gh_i(x_e(t_k),x_i(t_k))u_e(t_k)\\
\quad +M_{i,k}+\beta_{i}(h_{i}(x_e(t_k),x_i(t_k)))\\
\quad +[L_fh_i(x_e,x_i)-L_fh_i(x_e(t_k),x_i(t_k))\\
\quad +(L_gh_i(x_e,x_i)-L_gh_i(x_e(t_k),x_i(t_k)))u_e(t_k)\\
\quad +\beta_{i}(h_{i}(x_e,x_i))-\beta_{i}(h_{i}(x_e(t_k),x_i(t_k)))]-\beta_{i}(h_{i}(x_e,x_i))\\
\stackrel{\text{\eqref{Sampling-data CBF}}}{\geq}  \phi_{i,k}-(c_{f,i}+c_{g,i}u_{e,\max}+c_{\beta,i})(\|x_e-x_e(t_k)\|\\
\quad +\|x_i-x_i(t_k)\|)-\beta_{i}(h_{i}(x_e,x_i))\\
\geq -\beta_{i}(h_{i}(x_e,x_i)).
\end{array}
\end{align*}
The argument $t$ in $x_e(t)$, $x_i(t)$, and $\dot{x}_i(t)$ is omitted for brevity. Thus, under the CBF-ZOH control satisfying \eqref{Sampling-data CBF}, if $\dot{x}_i(t_k)\in C_i(t_k)$, we obtain $\dot{h}_i(x_e,x_i)\geq-\beta_i(h_i(x_e,x_i))$ for all $t\in [t_k,t_{k+1})$ and all $k$. 
Finally, combining with Proposition \ref{prop1}, the conclusion holds.
\end{proof}

Based on Theorem \ref{theorem1}, the sampled-data CBF-based optimization problem is constructed for each sampling time $t_k$ (this argument is omitted below for brevity) as:
\begin{subequations}
\label{CBF-QP with inference}
\begin{align}
	&u_{e}(\mathbf{x})=\arg\min_{u_e} \|u_e-u^r_{e}\|^2 \\
        \nonumber
	\text{s.t.} \ \ &L_fh_i(x_e,x_i)+L_gh_i(x_e,x_i)u_e+M_{i,k}\\
 &\ \ \ \ \geq-\beta_{i}(h_{i}(x_e,x_i))+\phi_{i,k},\\
 &u_{e}\in \mathcal{U}_e,\end{align} 
\end{subequations}
where $i\in \mathcal{V}$, $u^r_{e}$ is the reference control for ego agent. To summarize, the control input $u_{e}$ is generated by solving the CBF-based optimization \eqref{CBF-QP with inference} at each sampling time $t_k$ and we apply zero-order hold (ZOH)  for $t\in [t_k,t_{k+1})$, which enables the ego agent to track its reference trajectory as closely as possible with guaranteed safety with high average confidence. 

\vspace{-5pt}
\section{Case Study}

The proposed method is implemented in simulation to demonstrate its effectiveness. The scenario involves  one ego agent and two other agents moving in a two-dimensional space. The objective of the ego agent is to track its reference trajectory while ensuring safety w.r.t. the other agents. Agent 1 is moving towards the ego agent. Agent 2 moves close to ego agent while going away from agent 1. The ego agent does not know the intents of the other agents, but it can measure the positions of other agents at each sampling time. 
Our inference algorithm is utilized to learn their dynamics. In the first step, offline datasets are collected by navigating the ego agent using a CBF-based controller that forces the agent to track various reference trajectories, such as circles, sine waves, and spirals, while remaining safe with respect to other agents. The two other agents move with single integrator model following their intentions in reaction to the ego agent. 

The ego agent dynamics are given as $\dot{p}_x=v\cos(\theta)$, $\dot{p}_y=v\sin(\theta)$, $\dot{\theta}=\omega$, where $p_x$, $p_y$ are its position coordinates, and $v$ and $\omega$ are its linear and angular velocities that serve as the control inputs.  Let $p_{x,i}$, $p_{y,i}$, $i=1,2$ denote the positions of two other agents. 
We use RBFNNs with 8 neurons to capture the behaviors of agents 1 and 2, respectively. The widths in Gaussian basis function are set as 0.85. By employing \texttt{newrb} in MATLAB neural network toolbox, the RBFNNs are trained offline using the collected datasets, where the weights and centers are optimized to minimize the estimation errors. 

We set the offline-obtained weights as initial values $\hat{W}_i(t_0)$ of the  online weights, and fix the offline-obtained centers $C_i$ of RBFNNs. Then, the update law \eqref{W_Update law} is utilized to tune the weights $\hat{W}_i$ online. 
We set the initial miscoverage/target failure probability as $\alpha_i(t_0)=\alpha_i=0.01$, and the learning rate as $\gamma=0.002$. The prediction sets corresponding to the RBFNN at time $t_k$ 
are obtained by adaptive conformal prediction in Section \ref{sec:acp}. For the ego agent, the CBF for the unicycle model with dynamic obstacles in \cite{wu2016safety} is chosen as safety sets and the control signal is obtained from the sampled-data CBF-based optimization problem in \eqref{CBF-QP with inference} at $t_k$ with ZOH.\footnote{Codes: \url{https://github.com/MrJUNHUIZHANG/CBF_NN_inference}.}
\begin{figure}[!htb]
	\begin{center}
 \vspace{-5pt}
\includegraphics[width=0.9\columnwidth]{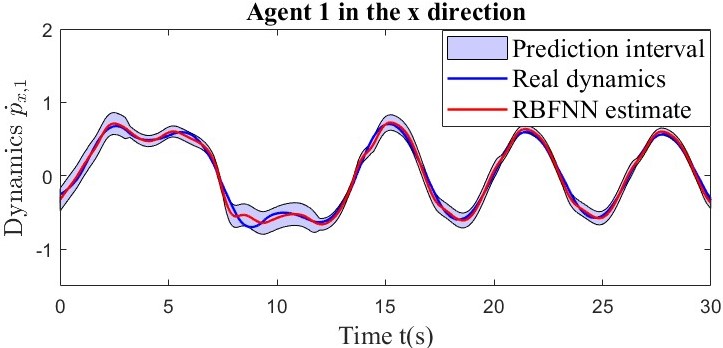}
        \hspace{200in}\vspace{-5mm}
		\caption{ Estimation of the dynamics of agent 1 in the $x$ direction.}
  \vspace{-15pt}
	\end{center}
 \end{figure}
 \begin{figure}[!htb]
	\begin{center}
\includegraphics[width=0.9\columnwidth]{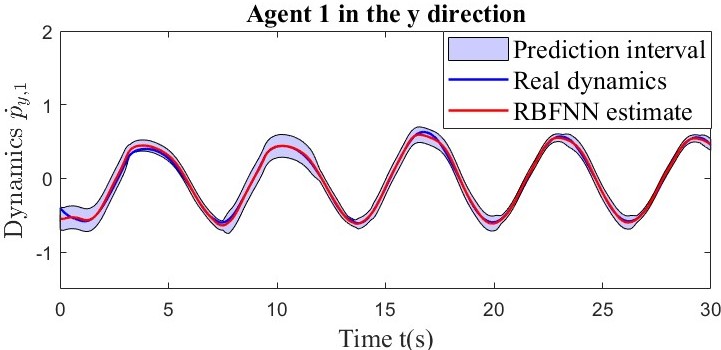}
        \hspace{200in}\vspace{-5mm}
		\caption{ Estimation of the dynamics of agent 1 in the $y$ direction.}
  \vspace{-10pt}
	\end{center}
 \end{figure}
\begin{figure}[!htb]
    \begin{center}
        \vspace{-5pt}
\includegraphics[width=0.9\columnwidth]{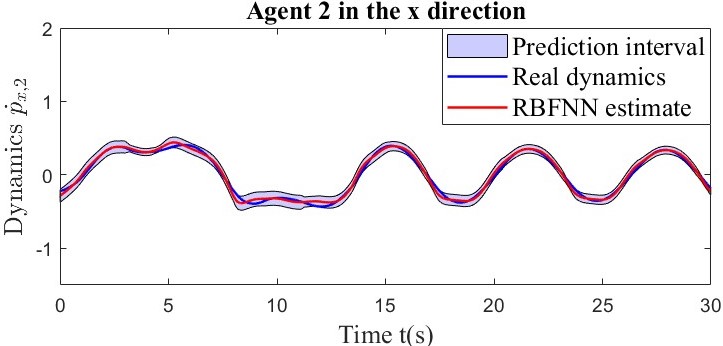}
        \hspace{100in}\vspace{-5mm}
        \caption{Estimation of the dynamics of agent 2 in the $x$ direction.}
         \vspace{-5pt}
    \end{center}
\end{figure}
\begin{figure}[!htb]
    \begin{center}
        \vspace{-10pt}
\includegraphics[width=0.9\columnwidth]{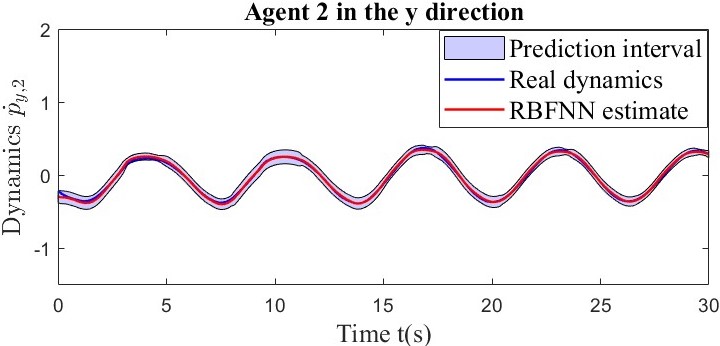}
        \hspace{100in}\vspace{-5mm}
        \caption{Estimation of the dynamics of agent 2 in the $y$ direction.}
         \vspace{-8pt}
    \end{center}
\end{figure}

\begin{figure}[!htb]
\begin{center}	
   \vspace{-10pt}
  \includegraphics[width=0.9\columnwidth]{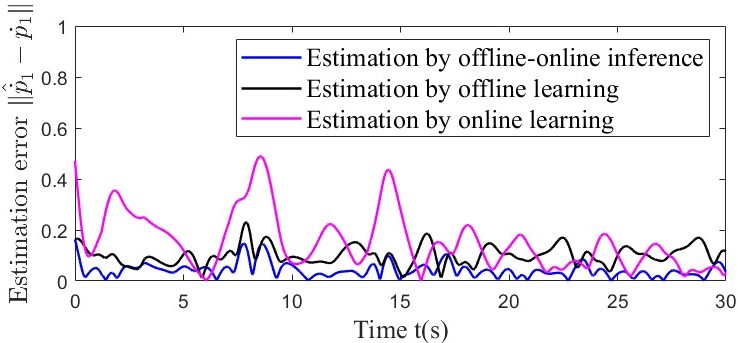}
		\hspace{100in}\vspace{-5mm}
		\caption{ Comparison of estimation performance for the dynamics of agent 1 (proposed vs offline only vs online only).}
  \vspace{-8pt}
	\end{center}
\end{figure}
\begin{figure}[!htb]
\begin{center}	
   \vspace{-10pt}
  \includegraphics[width=0.9\columnwidth]{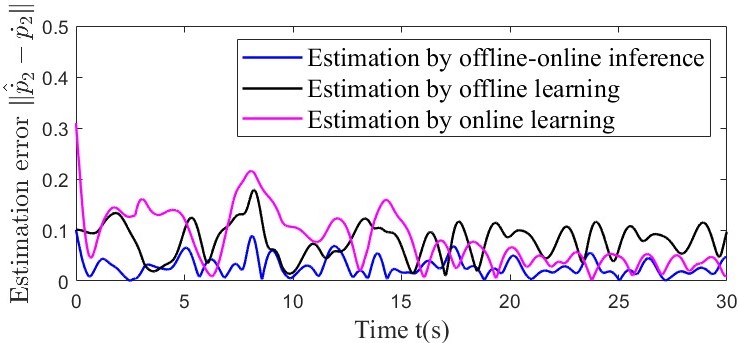}
		\hspace{100in}\vspace{-5mm}
		\caption{Comparison of estimation performance for dynamics of agent 2 (proposed vs offline only vs online only)}
  \vspace{-5pt}
	\end{center}
\end{figure}
\begin{figure}[!htb]
	\begin{center}		 
 \includegraphics[width=0.95\columnwidth]{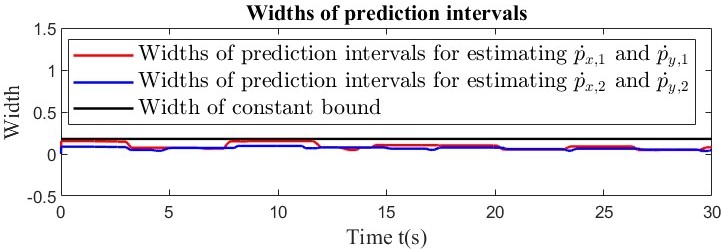}
		\hspace{200in}\vspace{-5mm}
		\caption{Widths of prediction sets for agents 1 and 2 using adaptive conformal prediction.}
   \vspace{-15pt}
	\end{center}
\end{figure}
The estimates of other agents' dynamics during online process are depicted in Figs. 2--5. 
We observe that the prediction sets cover the real values at all times, ensuring safety through the CBF-based control that incorporates prediction sets. 
We also compare the performance of the offline-online inference method via the estimation error against purely offline and online learning approaches, see Figs. 6--7. 

The widths of the prediction sets are shown in Fig. 8. Compared to earlier work in \cite{parwana2022trust}, which utilizes a conservative constant bound and does not attempt to infer other agents' dynamics, our proposed method imposes less conservatism in the CBF-based conditions. Finally, the trajectories of the ego agent, agents 1 and 2  are depicted in Fig. 9, where the colors transitioning from dark to light represent the trajectories of the agents over time. The reference trajectory of the ego agent violates the safety constraint. Nevertheless, the ego agent minimizes the distance between its trajectory and the reference trajectory while ensuring safety, ultimately maintaining a distance of 1.1 from the agents 1 and 2. Collision is avoided during the entire process. In addition, we ran 20 simulation trials with different initial positions, and observed that collisions did not occur, which is well within the fixed target failure probability of $\alpha_i=0.01$. This  demonstrates the effectiveness of our method.
\begin{figure}[!htb]
	\begin{center}	
  \vspace{-5pt}
  \includegraphics[width=0.9\columnwidth]{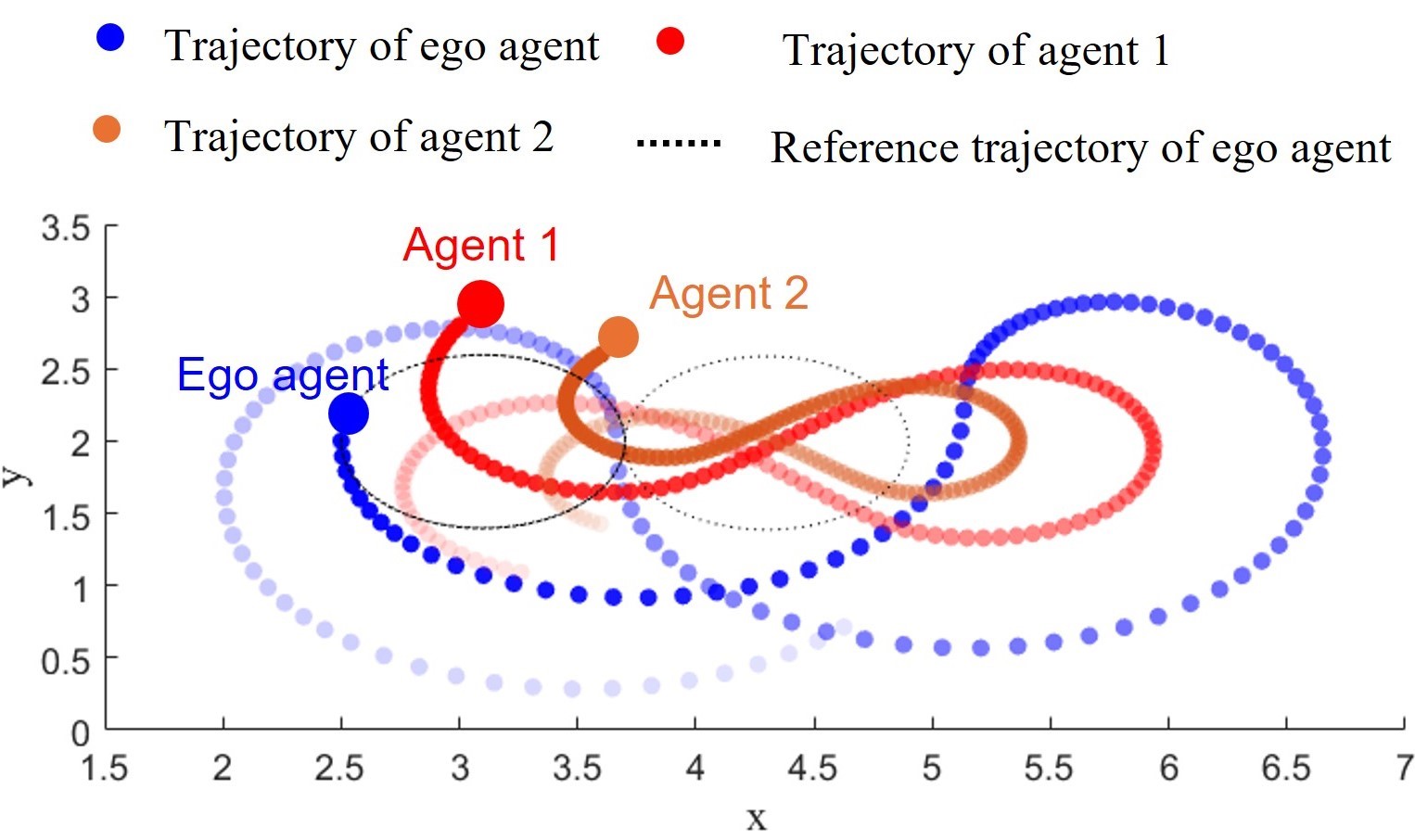}
\hspace{200in}\vspace{-5mm}\caption{Trajectories of ego agent, and agents 1 and 2.}
 \vspace{-10pt}
	\end{center}
\end{figure}

\begin{figure}[!htb]
	\begin{center}
		\vspace{-10pt}
  \includegraphics[width=0.9\columnwidth]{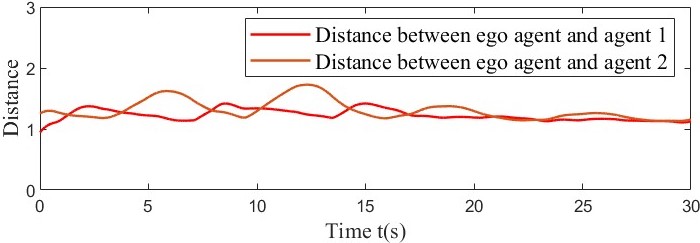}
		\hspace{200in}\vspace{-5mm}
		\caption{Distances between ego agent and agents 1 and 2.}
  \vspace{-19pt}
	\end{center}
\end{figure}

\vspace{-5pt}
\section{Conclusions}

This paper presented an inference method for an ego agent to estimate the dynamics of other agents, and its integration within a safe control synthesis framework. Offline-trained RBFNNs are utilized to learn the dynamics of other agents, and their weights are updated online with new observations. Adaptive conformal prediction is utilized to quantify the estimation error of the RBFNNs for other agents' dynamics by generating prediction sets covering the true value with desired average confidence level. Subsequently, a sampled-data CBF-based optimization problem is formulated and solved to guarantee safety with the desired average confidence level by incorporating the prediction set. 
Finally, a case study  demonstrated the effectiveness of the proposed method. Future work includes extending the methods to distributed scenarios involving cooperative/non-cooperative agents.

\section{Acknowledgements}
We would like to thank Hardik Parwana and Taekyung Kim for inspiring and insightful discussions on this work.

\bibliographystyle{IEEEtran}
\bibliography{bibligraphy}
\end{document}